\documentclass{llncs}

\usepackage{enumitem}
\usepackage{wrapfig}
\usepackage{multirow}
\usepackage{amsmath}
\usepackage{amssymb,amsfonts, graphicx}
\usepackage{array}
\usepackage{fourier}
\usepackage{xspace}
\usepackage{cancel}

\usepackage{colortbl}
\usepackage{figlatex}

\usepackage{algorithm} 
\usepackage[noend]{algpseudocode}

\usepackage{url}
\usepackage{flushend}

\usepackage{color}

\definecolor{taskyblue}{rgb}{0.44706, 0.56078, 0.81176}         
\definecolor{ta2skyblue}{rgb}{0.20392, 0.39608, 0.64314}        
\definecolor{ta3skyblue}{rgb}{0.12549, 0.29020, 0.52941}        

\definecolor{ta4skyblue}{rgb}{0.54706, 0.66078, 0.1176}         

\usepackage{cite}
\usepackage{times}

\usepackage{comment}
\usepackage{subfig}

\newcommand{\prog}[1]{\textnormal{\scshape#1}}

\usepackage{tikz}
\usepackage{pgfplots}

\usetikzlibrary{arrows,decorations,shapes,shadows}

\tikzset{
  ashadow/.style={opacity=.25, shadow xshift=0.07, shadow yshift=-0.07},
}

\tikzstyle{lightgray}=[fill=black!5]
\tikzstyle{gray}=[fill=black!20]
\tikzstyle{darkgray} = [fill=black!35]

\tikzstyle{lightblue} = [fill=blue!5]
\tikzstyle{lightred} = [fill=red!5]

\tikzstyle{ord}=[rounded rectangle, draw,minimum size=15pt,inner sep=1pt] 
\tikzstyle{goal}=[rounded rectangle, very thick,draw,minimum size=15pt,inner sep=1pt, fill=white] 

\tikzstyle{trans}=[>=angle 90,->,line width=0.02cm]
\tikzstyle{player}=[rectangle,fill=black,minimum size=3pt,inner sep=0pt]
\tikzstyle{distribution}=[circle,fill=black,minimum size=3pt,inner sep=0pt]

\tikzstyle{transparent} = [fill opacity=0.05]

\newcommand{\init}{\ensuremath{\mathit{init}}}
\newcommand{\program}{\ensuremath{\mathcal P}}

\newcommand{\transby}[1]{\overset{#1} E}
\newcommand{\transbystar}[1]{\overset{#1}{E^*}}

\newcommand{\pathsequence}
{\ensuremath{
	(\locvec_0,a_{0,\bar{l}(0)},\locvec_1) \ldots 
	(\locvec_{N-1},a_{m,\bar{l}(m)},\locvec_N)
	}}

\newcommand{\independent}[2]{#1\ ||\ #2}
\newcommand{\dependent}[2]{#1\ \cancel{||}\ #2}
\newcommand{\hb}[1]{\ensuremath{\rightarrow_{#1}}}

\newcommand{\thread}{\ensuremath{T}}
\newcommand{\threads}{\ensuremath{\mathcal T}}

\newcommand{\locs}{\ensuremath{L}}

\newcommand{\locvec}{\ensuremath{\bar{l}}}

\newcommand{\locvecs}{\ensuremath{L_G}}

\newcommand{\actions}{\ensuremath{A}}

\newcommand{\wff}{\ensuremath{\mathcal F}}

\newcommand{\art}{\ensuremath{\mathcal A}}

\newcommand{\covers}{\ensuremath{\sqsubseteq}}

\newcommand{\annotation}{\ensuremath{\phi}}

\newcommand{\formula}{\ensuremath{\mathcal F}}

\newcommand{\races}{\ensuremath{\lessdot_{\pi}}}

\begin{document}

\title{ {\em AbPress}: Flexing Partial-Order Reduction\\  and Abstraction}

\author{Daniel Kroening \and Subodh Sharma \and Bj\"orn Wachter}
\institute{Department of Computer Science, University of Oxford, UK}

\maketitle

\begin{abstract}
%


Partial-order reduction (POR) and lazy abstraction with interpolants are two
complementary techniques that have been successfully employed to
make model checking tools for concurrent programs effective.

In this work, we present {\em AbPress}
-- {\bf A}bstraction-{\bf b}ased {\bf P}artial-order {\bf R}eduction with
{\bf S}ource-{\bf S}ets -- 
an algorithm that fuses 
a recently proposed and powerful dynamic POR technique
 based on {\em source-sets} and lazy abstraction to obtain
 an efficient
software model checker for multi-threaded programs.
It trims the interleaving space 
by taking the abstraction and {\em source-sets} into account.
We amplify the effectiveness of {\em AbPress} with a novel 
solution that summarizes the accesses to shared variables over
a collection of interleavings.

We have implemented {\em AbPress} in a tool that analyzes concurrent
programs using lazy abstraction, {\em viz.,}\, Impara.
Our evaluation on the effectiveness of the presented approach has
been encouraging. {\em AbPress} compares favorably to existing 
state-of-the-art tools in the landscape.

\end{abstract}

\section{Introduction}
The generation of safety proofs for concurrent programs remains a major challenge.
While there exist software verification tools based on abstraction that scale
to sequential systems code~\cite{DBLP:journals/cacm/BallLR11}, that
cannot be said for multi-threaded software.
Abstraction-based verification of sequential programs works by annotating program locations in the
control-flow graph with safety invariants. 
However, applying  similar techniques to concurrent software
is ineffective, as interleavings lead to an explosion of the control-flow graph.  
Therefore, along with abstraction of data,
techniques are needed to effectively deal with
interleaving explosion.  
Partial-order reduction (POR) \cite{Pel93,GW91,valmari91}, a 
path-based exploration approach, 
is a technique that addresses the explosion in the interleaving space. 
The key notion in POR techniques is the {\em independence} of actions.
Independent actions can commute, resulting in interleavings that cause no observable change in the 
output. All interleavings obtained by commuting independent actions fall
into the same equivalence class. Thus, exploring only representative interleavings results
in a reduction in the number of interleavings explored in total.
Consider the example in Figure~\ref{fig:improving-impara}. The first two steps of 
$T_1$ are independent with the first two steps of $T_2$ (since they write to different 
shared variables). Out of the two interleavings
shown in Figure~\ref{fig:covers-por}, POR will identify that
the interleaving $\langle 00,10,11,21 \rangle$ does not need to be explored.

There are also path-based techniques to address the problem of
data-state explosion. A prominent technique is lazy abstraction
with interpolants (the Impact algorithm)~\cite{DBLP:conf/cav/McMillan06}. 
The Impact approach begins by unwinding a program's control-flow graph
into a tree.  Each node in the tree (encoding the control location) is
initially labeled with the state predicate \texttt{TRUE}, which indicates
reachability of the node from the initial location.
On reaching an error location, the node labels along the path to
that node are updated with interpolants in order to prove that the error 
state is unreachable. The starting node is labeled \texttt{TRUE} and 
each subsequent node is assigned a formula that implies the
next node's formula by executing the intervening program instruction. 
If the error node is labeled with \texttt{FALSE} then the approach has
proved the path to be infeasible.
The path exploration can terminate early. This happens when Impact discovers
{\em covered} nodes. When two nodes $v_1$ and $v_2$ in the abstract reachability tree 
have the same program control location and the invariant at $v_1$ subsumes the invariant at
$v_2$, then we say that $v_1$ {\em covers} $v_2$. This implies that it is no longer 
necessary to explore the reachability tree that follows $v_2$.
For instance, in Figure~\ref{fig:covers-por} location $21$ in the
interleaving $\langle 00,10,11,21 \rangle$ has the same interpolant as the
interpolant in location $21$ of interleaving $\langle 00,10,20,21,22,32
\rangle$.  Thus, the right $21$ node is covered by the left $21$ node. 
Observe that any implementation of POR would have eagerly detected the
independence between $t_1: x = x+1$ and $t_2: y =-1$ and the exploration of
the right interleaving would have been avoided.  The notion of covers
is thus most useful when control-flow branching is present in the program.
\begin{figure}[t]
\begin{minipage}[b]{0.5\textwidth}
    \centering{
    \scriptsize
    \begin{tabular}{l}
    \begin{tabular}{|l|l|l|}\hline
    $Main$              & $\thread_1$             & $\thread_2$\\ \hline
    \texttt{x=0; y=0;}         &  \texttt{1: x =1;}  & \texttt{1: y=-1;}\\
     \texttt{create}($\thread_1$);         &  \texttt{2: x=x+1;}   & \texttt{2: y= y+1;}  \\ 
    \texttt{create}($\thread_2$);&   \texttt{3:  x=y;}          &      \\
    \texttt{join}($\thread_1$);             &                         &     \\ 
    \texttt{join}($\thread_2$);             &                         &     \\  
    \texttt{assert ($ y \geq 0$)};             &                         &     \\ \hline 
    \end{tabular}
    \end{tabular}}
    \caption{Example with races}
    \label{fig:improving-impara}
\end{minipage}
\hspace{2em}
\begin{minipage}[b]{0.3\textwidth}
\centering{\includegraphics{fig/covers-por.fig}}
\caption{Covers}
\label{fig:covers-por}
\end{minipage}

\vspace{-.5cm}
\end{figure}

Both Impact and POR,
in particular dynamic POR (DPOR)~\cite{Flanagan:2005,Abdulla:2014},
use backtracking mechanisms to explore alternative choices at 
control locations: Impact uses backtracking for branching control
flow and DPOR for interleavings.
Due to the operational similarity and respective effectiveness in addressing
problems arising from data and schedule explosion, Impact and DPOR are ideal
candidates to be fused.  Impara~\cite{wko2013} offers a framework where
Impact can be combined with a POR technique of choice.  Impara comes with an
implementation of the \emph{Peephole POR} (PPOR)
algorithm~\cite{Wang:2008:PPO:1792734.1792772}, which leaves room
for further improvements.
In particular, PPOR is known to be suboptimal for programs with more than
two threads.  Further, PPOR does not integrate a backtracking mechanism;
it is a symbolic approach where chains of {\em dependent} actions have to be
maintained at each node by ascertaining information from the future
execution of the program.  DPOR algorithms can potentially be more efficient
than PPOR.  The recent work in~\cite{Abdulla:2014} offers us an opportunity
to use a dynamically constructed set of {\em dependent} actions, namely {\em
source-sets}.
Opportunities also exist to fine-tune the fusion of Impact and DPOR
where the abstraction constructed by Impact feeds information into DPOR.
%

%

%
%

%
%
%
In this paper, we present a new verification algorithm for multi-threaded
programs where Impact and DPOR with source-sets are combined in a novel way.
Note that combining {\em covers} and DPOR in a sound manner is a non-trivial exercise. 
Consider Figure \ref{fig:cover-bt}. Let $n_C$ be the covering node and 
$n_c$ be the covered node.  
Let $n$ be the least common ancestor of nodes $n_c$ and $n_C$.

\begin{wrapfigure}{r}{0.45\textwidth}
\vspace{-2em}
\centering
\includegraphics[scale=0.75]{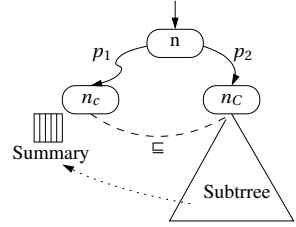}
\caption{Shared access summarization}
\label{fig:cover-bt}
\vspace{-2em}
\end{wrapfigure}
To discover alternate schedule choices in $p_1$, DPOR will first
enumerate the paths in the subtree from $n_C$.  For each path $p_s$ in the
subtree, DPOR will perform a dependence analysis for each step in $p_1$ with
each step in $p_s$ (backtracking mechanism).  Such an approach turns out
to be prohibitively expensive. Therefore we \emph{summarize} the accesses
in the subtree and re-use the summary.
This summarization technique is one key element to obtain an effective
combination of DPOR and the covers Impact uses.

\textbf{Contributions}:
Our main contributions are: 
(1)~an algorithm, {\em AbPress}, that combines source-set based DPOR with Impact, 
(2)~abstract summaries of shared variable accesses in a subtree to create a sound fusion of DPOR and 
covers, and
(3)~a comparison of {\em AbPress} with the state-of-the-art tools in the landscape. 
We present the basic definitions associated with Impact in Section~\ref{sec:prelims}. 
We present the essentials of source-set DPOR in Section~\ref{sec:sdpor} and 
abstract summaries in Section~\ref{sec:summaries}. The complete algorithm 
{\em AbPress} is presented in Section~\ref{sec:algo}. Experimental results 
are discussed in Section~\ref{sec:results}.

\section{Preliminaries}
\label{sec:prelims}

We consider a concurrent program $\program$ that is composed of a finite set
of threads $\threads$.  Each thread executes a sequence of operations given
in C or C++.  The threads communicate with each other by performing
operations on shared communication objects such as global variables,
semaphores and locks.
%
%
%
We only consider programs with a fixed number of threads.
A thread $T \in \threads$ is a four-tuple $T = (\locs, l_{0}, \actions,
l_{\varepsilon})$ consisting of a finite set of program control locations
$\locs$, an initial location $l_{0} \in \locs$, a set of thread actions
$\actions$ and an error location $l_{\varepsilon} \in \locs$.  A thread
action $a$ is a triple $a = (l, c, l')$ where $l, l' \in \locs$ are the
entry and exit program locations for the action, respectively, and $c$ is
the program instruction.  For brevity, we denote an action of thread $T$
that is enabled at location $l$ by $a_{T,l}$.  We assume that we are working
with an intermediate program representation where an instruction is either
an assignment or an assume statement.

For notational convenience, we identify instructions using their standard
formalisation as first-order formulae over the alphabet of primed and
unprimed program variables $V \cup V'$.  We denote the set of all such
formulae by $\wff(V \cup V')$.
Consider the example in Fig.~\ref{fig:improving-impara}. 
For the assignment $z=1$ in $T_1$, we have the action 
$(l_0,\,(z=2 \wedge z'=z),\,l_1)$.
%

A \emph{global control location} is a tuple with one component per thread,
and is given as function $\locvec: \mathcal{T} \rightarrow L$.  Let
$\locvecs$ be the set of all global control locations.
%
%
By $\bar{l}[T\mapsto l]$, we denote the global location where the location
of thread $T$ maps to $l$ while the locations of all the other threads
remain unchanged.  An action $a \in \actions$ from thread $T$ is enabled if
the action is enabled at $\locvec(T)$.



A program path $\pi$ is a sequence $\pi = \sigma_0, \ldots, \sigma_N$ where
$\sigma_i = (\bar{l}_i,T_i, a_i,\bar{l}_{i+1})$ consists of an action $a_i$
from thread $T_i \in \threads$ and $a_i$'s entry and exit global program
locations, $\bar{l}_i$ and $\bar{l}_{i+1}$.
A path is an \emph{error path} if $\bar{l}_0$ is initial control location
for all threads, and $\bar{l}_{N+1}$ contains an error location of a thread.

We denote by $\formula(\pi)$ the sequence of transition formulas
$\init^{(0)} \wedge R_0^{(0)}, \ldots R_{N}^{(N)}$ obtained by shifting each
$R_i$ $i$ time frames into the future.  Each $R_i$ is a transition formula
for an action at location $\bar{l}_i$.  We say that $\pi$ is feasible if
$\bigwedge R_i^{(i)}$ is logically satisfiable.  A~solution for $\bigwedge
R_i^{(i)}$ corresponds to a program execution and assigns values to the
program variables at each execution step.  The program is said to be
\emph{safe} if all error paths are infeasible.

\subsection{Interpolants, Invariants and ARTs}

In case a path is infeasible, an explanation can be extracted in the form of
an interpolant.   To this end, we recall the definition of \emph{sequent
interpolants}~\cite{DBLP:journals/tcs/McMillan05}.  A sequent interpolant
for formulas $A_1,\ldots,A_N$ is a sequence $\widehat A_1,\ldots, \widehat
A_N$ where the first formula is equivalent to true $\widehat A_1\equiv
\mathit{True}$, the last formula is equivalent to false $\widehat A_N\equiv
\mathit{False}$, consecutive formulas imply each other, i.e., for all
$i\in\{1,\ldots,N\}$, $\widehat A_{i-1}\wedge A_i \Rightarrow \widehat A_i$,
and the $i$-th sequent is a formula over the common symbols of its prefix
and postfix, i.e., for all $i\in\{1,\ldots N\}$, $\widehat A_{i}\in
\wff(A_1,\ldots,A_i) \cap \wff(A_{i+1},\ldots,A_N)$.  For certain
theories, quantifier-free interpolants can be generated for inconsistent,
quantifier-free sequences~\cite{DBLP:journals/tcs/McMillan05}.

An \emph{inductive invariant} is a mapping $I:  \locvecs \to \wff(V)$ such
that $\init \Rightarrow I(\locvec^i)$ (where $\locvec^i$ is the initial
global control location) and for all locations $\locvec \in \locvecs$, all
threads $\thread\in\threads$, and actions $a=(l,R,l')\in \thread$ enabled in
$\locvec$, we have $I(\locvec)\wedge R \Rightarrow I(\locvec[\thread\mapsto
l'])$.  A~\emph{safety invariant} is an inductive invariant with
$I(\locvec)\equiv \mathit{False}$ for all error locations $\locvec$.  If
there is a safety invariant the program is safe.

\begin{definition}[ART]
    An \emph{abstract reachability tree} (ART) $\art$ for program $\program$ is a tuple
    $(N,r,\transby{},\covers)$ consisting of a tree with nodes $N$, root node
    $r \in N$, edges $\transby{} \subseteq N \times \threads \times \mathcal{F}(V \cup V') \times N$, 
    and a covering relation $\covers
    \subseteq N^2$ between tree nodes such that:
    \begin{itemize}
      \item every nodes $n\in N$ is labeled with a tuple $(\bar{l},\annotation)$ consisting of
      a current global control location $\bar{l}$, and a state formula
      $\annotation$.
      We write $\bar{l}(n)$ and $\annotation(n)$ 
      to denote the control location and annotation, respectively, of node $n$.

      \item edges correspond to program actions, and
      tree branching represents both branching in the control flow
      within a thread and thread interleaving.  Formally, an 
      edge is a tuple $(v,\thread,R,w)$ where $v,w \in N$,
      $\thread\in\threads$, and $R$ the transition
      constraint of the corresponding action.
    \end{itemize} 
\end{definition}
We write $v\xrightarrow{\thread,R} w$ if there exists an edge
$(v,\thread,R,w)\in \transby{}$.  We denote $ v \leadsto w$ if there is a
path from $v$ to $w$ in $\art$. 
%
The role of the covering relation is crucial when proving program
correctness for unbounded executions.  It serves as an important criterion
in pruning the ART without missing error paths.
The node labels, intuitively, represent inductive invariants that represent
an over-approximation of a set of states.  {\em Covering} relation, in other
words, is the equivalent of a subset relation over this over-approximation
between nodes.  Suppose that two nodes $v,w$ share the same control
location, and $\annotation(v)$ implies $\annotation(w)$, {\em i.e.}, $v
\sqsubseteq w$.  If there was a feasible error path from~$v$, there would be
a feasible error path from~$w$.  Therefore, if we can find a safety
invariant for $w$, we do not need to explore successors of $v$, as
$\annotation(v)$ is at least as strong as the already sufficient invariant
$\annotation(w)$.
Therefore, if $w$ is safe, all nodes in the subtree rooted in $v$ are safe
as well.  A node is covered if and only if the node itself or any of its
ancestors has a label implied by another node's label at the same control
location.

To obtain a proof from an ART, the ART needs to fulfil certain conditions,
summarized in the following definition:

\begin{center}
  \begin{definition}[Safe ART]
    Let $\art=(V,\epsilon,\transby{},\covers)$ be an ART.
    $\art$ is \emph{well-labeled} if the labeling is
     inductive, i.e., 
          $\forall (v,\thread,R,w) \in \transby {}:\ \locvec(v) = \locvec(w) \wedge \annotation(v) \wedge R 
           \Rightarrow \annotation(w)'$
          and compatible with covering, i.e., 
          $(v,w)\in\covers:\ \annotation(v)\Rightarrow \annotation(w)$ and
          $w$ not covered.
    $\art$ is \emph{complete} if all of its nodes are covered,
    or have an out-going edge for every action that is enabled at $\locvec$.
    \item $\art$ is \emph{safe} if all error nodes are labeled with
    $\mathit{False}$.
   \end{definition}
\end{center}

\begin{theorem}
  If there is a safe, complete, well-labeled ART of program $\program$, then $\program$ is safe.
\end{theorem}

\subsection{Path correspondence in ART}


%
Let the set of program paths be $\Pi_{CFG}$.  A program path $\pi\in\Pi_{CFG}$ is
covered by $\art$ if there exists a corresponding sequence of nodes in the
$\Pi$ (denoting the set of paths in $\art$), where corresponding means that the nodes visits the same control
locations and takes the same actions.  In absence of covers, the matching
between control paths and sequences of nodes is straightforward.

However, a path of the ART may end in a covered node.  For example, consider
the path $\langle 00,10,11,21,22 \rangle$ in the control-flow graph of 
Figure~\ref{fig:covers-por}.  While prefix   $\langle 00,10,11,21$ can be matched by node sequence $\langle v_{00} v_{10} u_{11} u_{21} \rangle$, node $u_{21}$ is
covered by node $v_{21}$, formally $u_{21} \covers v_{21}$.  
We are stuck at node $u_{21}$, a leaf with no out-going
edges. In order to match the remainder of the path, our solution is to allow the corresponding sequence to ``climb up''
the covering order $\covers$ to a more abstract node, here we climb from
$u_{21}$ to $v_{21}$.  Node $v_{21}$ in turn must have a corresponding out-going
edge, as it cannot be covered and its control location is also $\locvec_2$. 
Finally, the corresponding node sequence for  $\langle 00,10,11,21,22 \rangle$ is 
$\langle v_{00} v_{10} u_{11} v_{21} v_{22} \rangle$.


This notion is formalized in the following definition:
\begin{center}
\begin{definition}[Corresponding paths \& path cover]
  \label{def:pathcover}
  Consider a program $\program$.  Let $\art$ be an ART for $\program$ and let
  $\pi=\pathsequence$ be a program path.  A \emph{corresponding path} for
  $\pi$ in $\art$ is a sequence $v_0,\ldots,v_n$ in $\art$ such that, for all
  $i\in\{0,\ldots,N-1\}$, $\locvec(v_i) = \locvec_i$, and
  \begin{align*}
    \exists u_{i+1} \in N:\ v_i, \xrightarrow{\thread_i, R_i} u_{i+1}, 
    a_i = (\locvec_i, R_i, \locvec_{i+1}) 
    \wedge (u_{i+1}=v_{i+1} \vee u_{i+1} \covers v_{i+1})
  \end{align*}
  A program path $\pi$ is covered by $\art$ if there exists a corresponding path $v_0,\ldots,v_n$ in $\art$.
\end{definition}
\end{center}

\begin{center}
\begin{proposition}
  Let $\program$ be a program.
  Let $\Pi$ be a representative set of program paths.
  Assume that $\art$ is safe, well-labeled and
  covers every path $\pi \in \Pi$. Then program $\program$ is safe.
  \label{proposition:picompleteness}
\end{proposition}
\end{center}



  


We denote the set of enabled actions from a node $n \in N$   by 
$\mathit{enabled}(n)$. The edge from node $n$ is denoted by $E(n)$.
  For any action $a$, let $\mathit{proc}(a) = T$ return 
the thread executing the action. We identify 
the unique successor node obtained after firing  $a$ from $n$ by $a(n)$.
In any given node  $n \in N$, let $\mathit{next}(n, T) = a_{T, \locvec(T)}$
denote the unique next action to be executed from thread $T$ 
after $n$.
For a path $ \pi \in \Pi$, the action fired from node $n \in \pi$ is $a_n$.

\section {Partial Order Reduction with Source-sets}
\label{sec:sdpor}
The basis for reduction using POR is the \emph{independence relation} among
concurrent actions.  Intuitively, two concurrent actions are independent if
executing then in any order leads to the same final state.  Thus, a path
$\pi'$ obtained by commuting adjacent independent actions in $\pi$ is same
in behavior as $\pi$.
The equivalence class representing all behaviorally similar interleavings is
commonly known as a {\em Mazurkiewicz
trace}~\cite{DBLP:conf/ac/Mazurkiewicz86}.
In other words, Mazurkiewicz traces represent the partial order among
the events of an execution path.
It suffices to explore only representative execution (or one linearization) of each Mazurkiewicz
trace.  In context of this work, it means that exploring 
representative paths in $\art$ will suffice.
\begin{center}
\begin{definition}[Independent actions]
Let $S$ represent the set of all execution states of the program.
Two actions $a_1$ and $a_2$ are {\em independent}, denoted
by $ \independent {a_1}{ a_2}$, iff 
the following conditions hold for all $s \in S$:
\begin{itemize}
\item Enabled: if $a_1$ is enabled in $s$ then 
   $a_2$ is enabled in $a_1(s)$ iff 
  $a_2$ is enabled in $s$ and
\item Commute: $a_1 (a_2(s) ) = a_2 (a_1 (s))$
\end{itemize}
\end{definition}
\end{center}
The definition of independence is impractical to implement (as it
requires a universally quantified check over the state-space).
In practice, easily-checkable conditions can be provided to determine
dependence of two actions (denoted by $\dependent{}{}$): for instance, two
actions that are concurrent at location $\locvec$ that acquire the same lock
or access the same shared variable (with one action performing a write) are
dependent.  In our setting, we consider actions that are enabled at a global
location $\locvec$ to be independent when they {\em commute}.
POR algorithms operate by first 
computing a subset of {\em relevant} enabled actions from a node and explore
only the computed subset from a scheduled node.
Some of the popular techniques to compute this subset 
are {\em persistent-set}  
and {\em sleep-set} techniques~\cite{DBLP:books/sp/Godefroid96}.
Briefly, a set $P$ of threads is persistent in a node if in any execution
from the node, the first step that is dependent with the first step of some
thread in $P$ must be taken by {\em some} thread in $P$.  Sleep-sets, on the
other hand, maintain, at each state, information about past explorations
and dependencies among transitions in the state in order to prune redundant
explorations from that state.  An elaborate exposition on these topics is
beyond the scope of this paper.  For a detailed discussion on these
techniques, refer~\cite{DBLP:books/sp/Godefroid96}.

Dynamic POR (DPOR) techniques~\cite{Flanagan:2005,yuyang-2008,Abdulla:2014}
compute the dependencies on the fly.  This leads to the construction of
more precise persistent-sets, thereby resulting in potentially smaller
state-graphs for exploration.
The central concept in most DPOR algorithms is that of a {\em race}.  DPOR
algorithms check whether actions in a path are racing and if found racing then the algorithm
tries to execute the program with a different schedule to revert the race.
We use $<_{\pi}$ to denote the total order among the nodes in the 
path $\pi \in \art$. Let $\hb{\pi}$ be the unique {\em happens-before}
relation over the nodes in the path $\pi \in \art$ such that $\hb{\pi} \subseteq <_{\pi}$.
Formally, consider  $u, v \in N$; if  $ u \hb{\pi} v$ then
$ u <_{\pi} v$ and $\dependent{a_u}{a_v}$.

\begin{definition}[Race]
Two actions $a_u$ and
$a_v$ from nodes $u$ and $v$ in a path $\pi \in \Pi$ are in a {\em race}, denoted by $ u \races v$,
 if the following conditions hold
true:(i) $ u \hb{\pi} v$ and $proc(a_u) \neq proc(a_v)$ and (ii) there does not
exist a node $w: u < w < v$ and  $ u \hb{\pi} w
\hb{\pi} v$.
\label{def:race}
\end{definition}
%

DPOR was first introduced with {\em persistent-sets}~\cite{Flanagan:2005}. 
However, recently in~\cite{Abdulla:2014}, an optimal
strategy to perform DPOR was presented.  Instead of using persistent-sets, the optimal
DPOR relies on a new construct, namely {\em source-sets}.
Succinctly, a source-set $S$ at a state $s$ is a set of threads that must be
explored from $s$ such that for each execution $E$ from $s$ there is some
thread $p \in S$ such that the first step in $E$ dependent with $p$ is by
$p$ {\em itself}.  Unlike persistent-sets where the first dependent step
with $p$ is taken by some thread in the set, in source-sets the first
dependent step with thread $p$ is taken by $p$ itself.  This subtle
difference can lead to smaller exploration choices from a state. 
Source-sets are persistent-sets but all persistent-sets are not source-sets.
DPOR based on source-sets has demonstrated considerable savings over basic
DPOR with persistent-sets~\cite{Abdulla:2014}.

We provide a brief demonstration illustrating the differences between
source-sets and persistent-sets using the example in
Figure~\ref{fig:odpor-ex} (borrowed from~\cite{Abdulla:2014}).  Consider the
path $r_1.r_2.q_1.q_2$ from the initial node and the persistent-set
$\{p,q\}$.  Note that $r_2$ is dependent with $p$ but thread $r$ is not in the
persistent-set.  By the preceding explanation 
of persistent-sets, it implies the
persistent-set at the initial node must also include $r$.
\begin{wrapfigure}{r}{0.4\textwidth}
\vspace{-1em}
\scriptsize 
\centering{\includegraphics[scale=1]{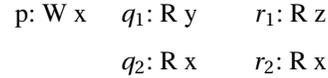}}
\caption{Example for Source-sets}
\label{fig:odpor-ex}
\vspace{-1em}
\end{wrapfigure}
Consider again the prefix $r_1.r_2.q_1.q_2$  from the initial node. Let the
source-set be $S:=\{p,q\}$.  The first step in the
prefix that is a dependent action with $p$  is $r_2$; however,
note that $r_2$ is mutually independent with the actions from the process
$q$.  Thus, by reordering, we obtain $r_1r_2q_1q_2 = q_1q_2r_1r_2$. 
According to the explanation of a source-set, it is now the case that the
first step in the execution prefix dependent with a source-set entry $q$ is
take by $q$ itself.  Thus, the given source-set $S$ is sufficient to explore
all executions starting from the start state.  By contrast, the
persistent-set definition mandated that $p, q,$ and $r$ are explored from the
start state.


For a path $\pi \in \Pi$ starting from node $n$, let $I_n(\pi)$ denote a
set of threads that have no {\em happens-before} predecessors in $\pi$. 
Intuitively, these are the ``first steps'' from threads $ p \in I_n(\pi)$ at nodes
$u \in \pi$.
That is,  there exists no $v \in \pi, v \neq u$  and
$v \hb{\pi} u$.  Let $WI_n(\pi)$ be the union of $I_{n}(\pi)$ and
the set of processes $p \in enabled(n)$ such for all actions $a$ in $\pi$,
we have $ \independent{next(n,p)}{a}$.  The set of threads $WI_n(\pi)$
represents the threads that can independently start an execution from the node $n$
covering all possible paths from $n$.

\begin{definition}[Source-sets]
A set $\prog{SSet}(n)$ is a source-set for the set of paths $\Pi$ after node $n$ if for 
each $ p \in \Pi$ we have $WI_n(p) \cap \prog{SSet}(n) \neq \emptyset $. 
\end{definition}
 
Our source-set based algorithm is similar to  Algorithm 1
in~\cite{Abdulla:2014}.  However, unlike the version in~\cite{Abdulla:2014},
our version of source-set DPOR operates in a symbolic execution engine. 
Procedure 
$\prog{ComputeBT}(u,v)$ in Algorithm~\ref{algo:abpress} calculates the source-set $\prog{SSet}$ at node $u$ when $a_u \races a_v$
incrementally. Procedure $\prog{NotDep(u,v)}$ is the sequence of nodes $\pi$ from the path $u \leadsto v$ (excluding $u$ 
and $v$) such that each node $w$ in the sequence is independent with $u$, {\em i.e.}, $u \nrightarrow_{\pi} w$.

\section {Summarization}
\label{sec:summaries}
Combining source-set DPOR and lazy abstraction in a naive manner can lead to unsoundness. 
Consider Figure~\ref{fig:cover-bt}. Impact with DPOR will 
explore $p_1$, compute the relevant backtrack choices for the steps within $p_2$, 
and finally stops exploring any further since $n_c \covers n_C$. 
However, a subset of paths following $n_C$ will also follow from the node~$n_c$.  
Terminating the dependency analysis  without considering
the dependencies among the shared variable accesses made in 
the sub-tree following $n_C$  will result in relevant backtrack
points in $p_2$ to be skipped. This is the source of unsoundness. 
In order to be sound, the DPOR algorithm must be invoked for each path suffix in the sub-tree that 
follows a covering node $n_C$ with each step in the prefix of the covered node $n_c$. 
Note that such a check quickly becomes expensive. 
We present an optimization of the above check by caching, for each shared variable, the set of threads that perform  
the ``earliest'' access to them. 
%

From before, an edge $ e = (u, T, a , w)$ shifts the control from node $u$ to node $w$ on action $a$.
Let signature of a node $Sig(e) = (t, R, W) $ be a tuple consisting of the owner thread,
the set of shared variables that is read by $a$ and the set of shared variables written by $a$. 
Let $\Pi$ be the set of paths starting from node $n$ to the final node,
 {i.e.}, for any path of the form $ n \leadsto w$  where $w$ is the final node with no actions enabled. 
%
\begin{definition}[Path Summary]
Let $\prog{sum}(p)$  be the signature of path $p = e.p'$ with $Sig(e) = (t, R, W)$ where the following
conditions hold:
\begin{itemize}
\item if $p'$ is empty then $\prog{sum}(p) = \{Sig(e)\}$
\item if exist $(t',R'W') \in \prog{sum}(p')$ such that $t = t'$, then 
$\prog{sum}(p) = \prog{sum}(p') \setminus \{(t',R',W')\}$ $\cup \{t, R\cup R', W \cup W'\}$
\item if exist $(t',Rd'Wr') \in \prog{sum}(p')$ such that $t \neq t'$ and $R \cap R' \neq \emptyset$ or
$W \cap W' \neq \emptyset$, then
 $\prog{sum}(p) = \prog{sum}(p')[(t',R',W') \mapsto (t',R'\setminus R,W'\setminus W)] \cup \{ Sig(e)\}$
\end{itemize}
\end{definition}




%
\begin{definition}[Node Summary]
The summary of a node $n \in \art$ is defined as the set 
\mbox{$\mathbb{S}(n) = \bigcup_{p \in \Pi}\prog{sum}(p)$} where $\Pi$ is the set of paths 
that start with root node $n$.
\end{definition}
%

\begin{theorem}[Soundness of Shared Access Summarization]
  Let  $\pi_1 =  u_1 \ldots u_n$ and
  \mbox{$ \pi_2 = v_1 \ldots  v_m$}  be two paths such that $ u_n \covers v_1$.
  For each node $u_i \in \pi_1$,  $\prog{SSet}(u_i)$ computed with 
  $\mathbb{S}(v_1)$  over-approximates  $\prog{SSet}(u_i)$ when 
  computed for the path $\pi_1.\pi_2$. 
  \label{theorem:summarization}
\end{theorem}
\begin{proof}
Assume that there exists a thread $t \in \prog{SSet}(u_i)$ when computed on $\pi_1.\pi_2$
 which is not 
present in $\prog{SSet}'(u_i)$ when computed with  $\mathbb{S}(v_1)$.
 Let the assumed entry be $(T, a)$ from node $v$. 
Since $v$ must race with node $u_i$,
 clearly $a_v$ ( must be the ``earliest'' action accessing the shared variables in a racing manner
after $a_{u_i}$ (from  Definition~\ref{def:race}). 
From the invariant of the constructive definition of $\prog{sum}(\pi_2)$, $a_v$ is a part of $\prog{sum}(\pi_2)$ 
and therefore a part of $\mathbb{S}(v_1)$.
This contradicts our assumption.
\end{proof}

We overload the operator for racing nodes; if $ u \races v$, then $a_u \races a_v$ and $Sig(E(u)) \races Sig(E(v)) $.
Consider Figure~\ref{fig:summaries}. Let nodes $v, w, z$ fire actions that have the earliest
accesses to variables $x,y$ in path $p_2$ and $p_3$, as shown in the figure. The summary at $n_C$ is 
$\mathbb{S}(n_C)=\{(t_2, \{y\},\{x\}), (t_3, \{\}, \{x\})\}$. Observe that $u \races w $ and 
$ u \races z $; therefore, we perform the source-set analysis for the path $u \ldots w$ and 
$u \ldots z$ by computing $\prog{SSet}(u)$.
\begin{wrapfigure}[13]{r}{0.5\textwidth}
\vspace{-2em}
\centering{\includegraphics{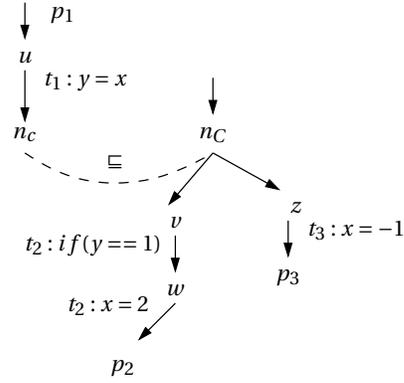}}
\caption{Source-sets with summaries}
\label{fig:summaries}
\vspace{-2em}
\end{wrapfigure}
Suppose we discover that $t_2 \in I_u(u \ldots w)$. We then add $t_2$ as an alternate schedule
choice to $\prog{SSet}(u)$. It is possible that $t_2$ at $u$ is disabled since there was 
no earlier node that updated the value of $y$ to one. This indicates that $\prog{SSet}(u)$
can potentially be overapproximate when computed with summaries.

\section{AbPress Algorithm}
\label{sec:algo}

\begin{algorithm*}[!t]
\vspace{-.25cm}
\begin{tabular}{p{0.45\textwidth}p{0.55\textwidth}}

\begin{algorithmic}[1]
\footnotesize
\Procedure {main}{$ $}
	\State $Q := \{ r\}$, $\covers := \emptyset$
	\While { $Q \neq \emptyset$}
		\State select and remove $v$ from $Q$
		\State $\prog{Close}(v)$
		\If {$v$ not covered}
			\If { $error(v)$ }
	 			\State $\prog{Refine}(v)$
			\EndIf
  		    \State $\prog{Expand}(v)$
                
		\EndIf
	\EndWhile
	\State
	\Return $\program$ is safe
\EndProcedure
  \item []
\Procedure {expand-thread}{$\thread, v$}
  
    \State $(\locvec,\annotation) := v$
    	\For { $(l,a_{T, \locvec(T)}, l') \in \actions(\thread)$ }
		\State $w:=$ fresh node 
    		\State $\locvec(w):=\locvec[\thread\mapsto l']$
      		\State $\annotation(w) := \mathit{True}$
      		\State $Q := Q \cup \{w\}$, $N := N \cup \{w\}$
      		\State $\transby {} := \transby {} \cup \{ (v,\thread,R,w) \}$   \label{algo:abpress:expand:trans}
   	\EndFor
\EndProcedure 
\item[]
\Procedure {Backtrack}{$v$}
\State $\pi :=r \ldots v $  path from $r$ to $v$
\For {$u, w \in \pi: u <_{\pi} w$}
\State compute  $\mathbb{S}(u)$
\If {$Sig(u) \races Sig(w)$}
\State $\prog{ComputeBT}(u,w)$
\EndIf
\If {exists $z: v \covers z$}
\For {$e \in \mathbb{S}(z) $ }
\If {$Sig(u) \races e$}
\State $\prog{ComputeBT}(u,g(e))$
\EndIf
\EndFor
\EndIf
\EndFor
\EndProcedure
\item[]
\Procedure {ComputeBT}{$u,v$}
\State $\pi' = r \ldots u$ path from $r$ to $u$
\State $\pi'' = \prog{NotDep}(u,v).v$
\If {$I_u(\pi'') \cap \prog{SSet}(u) = \emptyset$}
\State  $\exists t \in I_u(\pi''): \prog{SSet}(u) \cup= \{t\}$
\EndIf
\EndProcedure

\end{algorithmic}
&
\begin{algorithmic}[1]
\makeatletter\setcounter{ALG@line}{22}\makeatother
\footnotesize
\Procedure {expand}{$v$}
    \State $\thread := \prog{Choose}(v)$
    \If {$\thread=\bot$}
			\State $\prog{Backtrack}(v)	$
		\Else
			\State $\prog{Expand-thread}(\thread,v)$
		\EndIf
\EndProcedure 
\item []
\Procedure {close}{$v$}
  \For {$w \in \prog{Pre}(v) \wedge w$ uncovered $:v \sqsubseteq w$ }
  		\State $\covers := \covers \cup \{ (v,w) \}$
   		\State $\covers := \covers \setminus \{(x,y)\in \covers \mid v \leadsto y \}$
                
  \EndFor
  \If {v covered} 
  \State $\prog{Backtrack}(v)$
  \EndIf
  \State
\EndProcedure
\item[]
\Procedure {refine}{$v$} 
  \If {$\annotation(v) \equiv \mathit{False}$}
  	\State \Return
  \EndIf
  
	\State $\pi :=v_0,\ldots v_N $  path from $r$ to $v$
	\If {$A_0 \ldots A_N = \prog{ITP}(\formula(\pi))$}
			\For { $ i = 0 \ldots N$}
				\State $\annotation := A_i^{i}$
				\State $Q := Q \cup \{ w \mid w \covers v_i \}$
				\State $\covers := \covers \setminus \{ (w,v_i) \mid w \covers v_i \}$
				\State $\annotation(v_i) := \annotation(v_i) \wedge \annotation$
			\EndFor
      \For { $w \in V$ s.t. $w \leadsto v$}
      	\State $\prog{Close}(w)$
      \EndFor
	\Else
		\State \ abort (program unsafe)
	\EndIf

\EndProcedure 
\end{algorithmic}
\end{tabular}
\vspace{-.5cm}
\caption{AbPress}
\label{algo:abpress}
\end{algorithm*}

\prog{AbPress} is a combination of source-set DPOR with abstract summaries and Impact.
We give the pseudo-code in Algorithm~\ref{algo:abpress}.
A large part of  Algorithm~\ref{algo:abpress} is similar to Impara~\cite{wko2013}. 
Functions  $\prog{Backtrack(v)}$,  $\prog{Choose}(v)$  and  $\prog{ComputeBT}(u,v)$
are the contributions of this work. We now give an overview of the algorithm.
A work list $Q$ of nodes that are not fully explored is maintained along with 
the covering relation. Initially, $Q$ contains the root node $r$ and the 
cover relation is empty.
\prog{Expand} takes an uncovered leaf node and computes its successors.
     \prog{Choose} returns a thread that is chosen to be expored from a leaf node. 
     We do not provide the algorithm for \prog{Choose} but briefly summarize its functionality. If the 
     set of expanded threads and source-set from the node are empty, then any enabled
     thread is chosen, otherwise a thread from source-set is chosen. For every enabled action, it creates a fresh tree node $w$, and sets its location to the control
     successor $l'$ given by the action.  To ensure that the labeling is
     inductive, the formula $\annotation(w)$ is set to $\mathit{True}$.  Then
     the new node is added to the work list $Q$.  Finally, a tree edge is
     added (Line~\ref{algo:abpress:expand:trans}), which records the step from
     $v$ to $w$ and the transition formula $R$.  Note that if $w$ is an
     error location, the labeling is not safe; in which case, we need to
     refine the labeling, invoking operation \prog{Refine}.
          
\prog{Refine} takes an error node $v$ and, detects if the error path is feasible
     and, if not, restores a safe tree labeling.  First, it determines if
     the unique path $\pi$ from the initial node to $v$ is feasible by
     checking satisfiability of $\formula(\pi)$.  If $\formula(\pi)$ is
     satisfiable, the solution gives a counterexample in the form of a
     concrete error trace, showing that the program is unsafe. Otherwise, an
     interpolant is obtained, which is used to refine the labeling.  Note
     that strengthening the labeling may destroy the well-labeledness of the
     ART.  To recover it, pairs $w \covers v_i$ for strengthened nodes $v_i$
     are deleted from the relation, and the node $w$ is put into the work
     list again.

\prog{Close} takes a node $v$ and checks if $v$ can be added to the covering
     relation.  As potential candidates for pairs $v \covers w$, it only
     considers nodes created before $v$, denoted by the set $V^{\prec
     v}\subsetneq V$.  This is to ensure stable behavior, as covering in
     arbitrary order may uncover other nodes, which may not terminate.  Thus,
     only for uncovered nodes $w \in Pre(v)$, it is checked if
     $\locvec(w)=\locvec(v)$ and $\annotation(v)$ implies $\annotation(w)$. 
     If so, $(v,w)$ is added to the covering relation $\covers$.  To restore
     well-labeling, all pairs $(x,y)$ where $y$ is a descendant of $v$,
     denoted by $v\transbystar{} y$, are removed from $\covers$, as $v$ and
     all its descendants are covered. Finally, 
     if $v$ is covered by $z$, \prog{Backtrack} on $v$ is invoked. 
     The backtrack function performs the classic dependence analysis of DPOR.
     For each pair of nodes $u, w$ where $u,w$ in $r \ldots to v$ and $u$ races with $w$ we compute
     the source-sets by calling the function \prog{ComputeBT} (Lines 21-24). The functionality of
     \prog{ComputeBT} is responsible for computing source-sets and is similar to Algorithm 1 in \cite{Abdulla:2014}.
     Since $v$ is covered by  $z$, the \prog{Backtrack} function performs race analysis of each step $u$ in $r \ldots v$ 
     with each entry $e$ in the summary of $z$ (Lines 29-31). If $u$ and $e$ race 
     then the \prog{ComputeBT} function is invoked again (with a ghost node for $e$) 
     to compute the thread that should be added in the
     source-set.

\prog{Main} first initializes the queue with the initial node~$\epsilon$,
     and the relation $\covers$ with the empty set.  It then runs the main
     loop of the algorithm until $Q$ is empty, i.e., until the ART is
     complete, unless an error is found which exits the loop.  In the main
     loop, a node is selected from $Q$.  First, \prog{Close} is called to try
     and cover it.  If the node is not covered and it is an error node,
     \prog{Refine} is called.  Finally, the node is expanded, unless it was
     covered, and evicted from the work list.

\section{Experiments}
\label{sec:results}
\definecolor{taskyblue}{rgb}{0.44706, 0.56078, 0.81176}         
\definecolor{ta2skyblue}{rgb}{0.20392, 0.39608, 0.64314}        
\definecolor{ta3skyblue}{rgb}{0.12549, 0.29020, 0.52941}        

The purpose of our experiments is twofold: we would like to demonstrate the
effect of the techniques proposed in the paper, and evaluate the
competitiveness of our tool with comparable tools.  To this end, we compare
\prog{AbPress} (\prog{Impara} with Source-set DPOR and summaries) with
three different tools:
\begin{itemize}
  \item \prog{Threader}~\cite{DBLP:conf/cav/GuptaPR11}, a proof-generating
  software verifier for concurrent programs.  It is one of the few other
  tools that produce correctness proofs for concurrent programs.

  \item \prog{FMCAD'13}~\cite{wko2013}: \prog{Impara} with peephole
  partial-order reduction~\cite{Wang:2008:PPO:1792734.1792772}, which serves
  as a baseline to evaluate the benefit of partial-order reduction.

  \item \texttt{CBMC} (version 4.9)~\cite{cbmc}, to compare with bounded model
  checking.  Note that CBMC does not generate proofs for unbounded programs.
\end{itemize}
We evaluate on benchmarks of the Software Verification Competition~\cite{SVCOMP14}
(SV-COMP 2014) and on weak-memory Litmus tests (submitted to SV-COMP 2015):
\begin{itemize}
  \item \texttt{pthread}: This category contains basic concurrent data structures,
  and other lock-based algorithms.
  There are three challenging aspects to this category.
  (1)~The queue examples and the stack example contain arrays.
  (2)~The synthetic programs include the Fibonacci examples,
      which require a very high number of context switches
      to expose the bug.
  (3)~Some examples contain more than 10 threads.

  \item \texttt{pthread-atomic}: This category contains mutual-exclusion
  algorithms and basic lock functionality, which is implemented by
  busy-waits.  This creates challenging loop structures.  Some loops are
  unbounded, i.e., there exists no unwinding limit, and some loops are
  nested.
  
  \item \texttt{pthread-ext}: This category is primarily designed to test
  the capability of tools that can deal with a parametric number of threads,
  which we have indicated with $\infty$.  \prog{Impara} does not terminate
  without a thread bound in this case.  We ran \prog{Impara} with a thread
  bound of $5$, as this is the minimal number of threads it takes to expose
  all bugs.  This is the only category in which \prog{Impara} is incomplete,
  while tools that support parametric verification such as \prog{threader}
  have an advantage.
  
  \item \texttt{Litmus}: These are small programs that are used to detect
  weakenings of sequential consistency.  The benchmarks are \texttt{C}
  programs that have been instrumented to reflect weak-memory
  semantics~\cite{DBLP:conf/esop/AlglaveKNT13} by adding buffers.
  The high degree of nondeterminism makes them challenging to analyse.

\end{itemize}

\newcommand{\cp}{\checkmark}
\newcommand{\ca}{\checkmark}
\newcommand{\timeout}{\textbf{TO}}
\newcommand{\er}{\textbf{ERR}}

\begin{table*}
\centering
\resizebox{\textwidth}{!}{
\begin{tabular}{|l|c|c|*{2}r||*{2}r|*{3}r|*{3}r|}
\hline
\rowcolor{taskyblue!20}
\multicolumn{3}{|>{\columncolor{taskyblue!20}} c|}{}     
& \multicolumn{2}{>{\columncolor{taskyblue!20}} c|}{CBMC}
&\multicolumn{2}{>{\columncolor{taskyblue!20}} c|}{Threader} 
& \multicolumn{3}{>{\columncolor{taskyblue!20}} c|}{FMCAD'13} 
& \multicolumn{3}{>{\columncolor{taskyblue!20}} c|}{AbPress} \\ \hline
\rowcolor{taskyblue!20}
                         &  LOC/Threads   & safe &    & s       &       & s &    s & $|V|$ & SMT &  s & $|V|$& SMT \\
\hline
\multicolumn{13}{|c|}{SV-COMP -- pthread}\\
\hline
\texttt{queue\_ok\_true}    &  159/3    &  y & \cp          &    550.0   &   \er &    -- &    \timeout  & --  & --    & \bf 63.7 &  6489 & 14.7 \\\rowcolor{taskyblue!10}
\texttt{queue\_false}       &  169/3    &  n & \cp          &    9.3     &   \er &    -- &    \timeout  & -- &  -- & \bf 8.5 &  4867 & 1.5 \\\rowcolor{white}
\texttt{stack\_true}        &  120/3    &  y & \cp          &    230.0   &  \ca & 360.0   &   619   &  131507 & 336.5   & \bf 30.5  &  7875 & 17.2 \\ \rowcolor{taskyblue!10}  
\texttt{stack\_false}       &  120/3    &  n & \cp          &   \bf 0.5  &  \cp & 83.0  &    51    &18776       &  28.8  & 1.8 &  2366 &  0.1     \\ \rowcolor{white}  
\texttt{twostage\_3\_false} &  128/4    &  n & \cp          &    7.4     &   \cp &    760.0 &    15.8  &  4290 & 12.5   & \bf 1.2 &  3144 &  0.1 \\ \rowcolor{taskyblue!10}
\texttt{sync01\_true}       &  62/3     &  y & \cp          &    190.0   &   \cp &   \bf 0.2 & \bf 0.2  & 731  &  0.0  & 0.7 &  775 &  0.3  \\ \rowcolor{white}
\texttt{sigma\_false}       &  48/17    &  n & \ca          &    30      & \textbf{ERR} &    -- &    \timeout &   -- &  --  & \bf 3.5 &  692 &  0.1  \\ \rowcolor{taskyblue!10}
\texttt{indexer\_true}      &  83/14    &  y & \cp          &   \bf 1.4  &   \cp &    6.5 & \timeout &   -- &  --  & \textbf{ERR}  & --  & --   \\ \rowcolor{white} 
\texttt{reoder\_2\_false}   &  84/3     &  n & \cp          &    1.4     &   \cp &    2.4 & 60.1 & 11026   &  46.3  & \bf 0.9 & 1005  & 0.2   \\ \rowcolor{taskyblue!10} %
\texttt{reoder\_5\_false}   &  2866/6   &  n & \cp          &  \bf  1.4  &   \cp &    2.6 & \timeout &  -- &  --  & 63.7 & 14546  & 42.3   \\ \rowcolor{white} 
\texttt{lazy01\_false}      &  49/4     &  n & \cp          &    0.4     &   \cp &    4.5 & \timeout &   -- &  --  & \bf 0.1 & 147  & 0   \\ \rowcolor{taskyblue!10}  %
\texttt{bigshot\_p\_false}  &  34/3     &  n & \cp          &    0.3     &   \textbf{ERR} &    -- & 0.5 &  272  &  0.0  & \bf 0.1 & 147  & 0   \\  \rowcolor{white} %
\texttt{bigshot\_s\_false}  &  34/3     &  n & \cp          & \bf   0.4  &   \textbf{ERR} &    --& 0.5 &   324 &  0.5  & \textbf{WP} & 144  & 0.1   \\ \rowcolor{taskyblue!10} %
\texttt{bigshot\_s\_true}   &  34/3     &  y & \textbf{WA}  &    0.4     &   \textbf{ERR} &    -- & 0.7 &  324 &  0.5  & 0.1 & 147  & 0.1   \\ \rowcolor{white} %
\texttt{fib\_bench\_true}   &  43/3     &  y &  \cp & 17.0  &   \textbf{ERR} &    -- & \timeout &   -- &  --  & \textbf{TO} & -- &  --   \\ \rowcolor{taskyblue!10} %
\texttt{fib\_bench\_false}  &  40/3     & n &  \ca & \bf 1.0 &   \textbf{ERR} &    -- & \timeout &   -- &  --  & \textbf{TO} & --  & --   \\ \rowcolor{white} %
\hline
\multicolumn{13}{|c|}{SV-COMP -- pthread-atomic}\\
\hline
\rowcolor{taskyblue!10}
\texttt{scull\_true}          &  397/4    &  y & \cp  & \bf 5.4   &   \cp & 610   &    \timeout &  -- &  --  & 603.3  & 148,629  &  525  \\ \rowcolor{white}
\texttt{qrcu\_true}          &  147/3    &  y & \cp  &    850.0   &   \textbf{TO} &  -- &    \timeout &  -- &  --  & \bf 268.8 & 93742  &  219.7  \\ \rowcolor{taskyblue!10} %
\texttt{qrcu\_false}          &  147/4    &  n & \cp  &  \bf   0.5   &   \textbf{TO} &  --  &    0.9 &   1165 &  0.2  & 35.7 & 31453  &  19.1  \\ \rowcolor{white}
\texttt{dekker\_true}          &  54/3    &  y & \cp  &    120.0    &   \cp &  3.2   &    1.0 &   883 &  0.7  & \bf 0.1 &  331 &  0.0  \\ \rowcolor{taskyblue!10} %
\texttt{peterson\_true}        &  41/3    &  y & \cp  &   2.7    &   \cp &   5.3 &   0.6  &  746 & 0.4  & \bf 0.9 &  1832 &  0.5 \\ \rowcolor{white}
\texttt{lamport\_true}         &  75/3    &  y & \cp  &  850.0    &   \cp &  37 &    3.8 &  2560 &3.1 & \bf 1.1 & 3612 & 0.4 \\\rowcolor{taskyblue!10} 
\texttt{szymanski\_true}       &  54/3    &  y & \cp  &   7.4   &   \cp &    13.0 &\bf 1.2 &  1226 &   0.8 & 1.3 &  3098 &   0.7 \\\rowcolor{white}
\texttt{read\_write\_lock\_false}  &  51/5   &  n & \cp  & \bf 0.4    &   \cp &  22 &  5.3   & 4497 &  3.6 & 0.6 & 4899 &  53 \\ \rowcolor{taskyblue!10} %
\texttt{read\_write\_lock\_true}  &  51/5    &  y & \cp & \bf 0.8    &    \cp & 17.0  &    842 & 93073 & 770.3 & 66.9 & 66041 & 16.9 \\\rowcolor{white} 
\texttt{time\_var\_mutex}& 54/3    &  y & \cp  &   2.4     &   \cp &    2.6 &   \bf 0.5 &   1075 &   0.2 &  0.6 &  1196 &  0.2 \\ \rowcolor{white}
\hline
\multicolumn{13}{|c|}{SV-COMP -- pthread-ext}\\
\hline
\rowcolor{taskyblue!10}
\texttt{01\_inc\_true}          & 47/$\infty$       &  y & \cp  &  850.0  &   \cp & \bf 1.2  &   26.1  &   433 &    & 13.0  & 148,629  &  3.8  \\ \rowcolor{white}
\texttt{02\_inc\_true}          & 51/$\infty$      &  y & \cp  &  850.0  &   \cp & \bf 3.9    &    \timeout &   -- & --   & 44.3  & 93742  &  219.7  \\ \rowcolor{taskyblue!10} %
\texttt{03\_incdec\_true}       & 80/$\infty$      &  y & \cp  &  850.0  & \cp   & \bf 13.0    &  168.7   &  485808 &  47.3 & 123 & 31453  &  19.1  \\ \rowcolor{white}
\texttt{04\_incdec\_cas\_true}  & 99/$\infty$      &  y & \cp  &  850.0   &   \cp &  \bf 38.0   &  \timeout   &  -- & --  & 148.3 &  331 &  2  \\ \rowcolor{taskyblue!10} %
\texttt{05\_tas\_true}          & 57/$\infty$      &  y & \cp  &  550.0  &   \cp &   5.3 &  \timeout  & --  & --  & \bf 0.3 &  1832 &  0.5 \\ \rowcolor{white}
\texttt{06\_ticket\_true}       & 75/$\infty$      &  y & \cp  &  850.0  &   \cp &  \bf 0.8 & \timeout   & --  & --  & \textbf{TO} & -- &--  \\\rowcolor{taskyblue!10} 
\texttt{07\_rand\_true}         & 97/$\infty$      &   y & \cp  & 850.0  &   \cp &   4.7  & \timeout  & --  & --   & \bf 0.3 &  3098 &   0.7 \\\rowcolor{white}
\texttt{08\_rand\_case\_true}   & 123/$\infty$      &   y & \cp  & 850.0  &   \cp &  12.0   & \timeout & -- &  -- & \bf 0.2 &  3098 &   0.7 \\\rowcolor{taskyblue!10}
\texttt{09\_fmax\_sym\_true}    & 59/$\infty$      &   y & \cp  & 730.0  &   \cp &  \bf 13.0   & \timeout &  -- & --   & \bf \textbf{TO} &  -- &   -- \\\rowcolor{white}
\texttt{10\_fmax\_sym\_cas\_true}    & 69/$\infty$      &   y & \cp   & 420.0  &   \cp &  \bf 37.0  &  \timeout & --  & --   & \textbf{TO} &  -- &   -- \\\rowcolor{taskyblue!10}
\hline
\multicolumn{13}{|c|}{Weak-memory and Litmus tests}\\
\hline
\rowcolor{taskyblue!10}
\texttt{mix000\_tso\_false}           & 359/3 &   n & \cp  & \bf 1.1  &   \textbf{ERR} &     & 4.5 &  4126 &  2.5  & 2.9 & 3507 &   0.2 \\\rowcolor{white} 
\texttt{mix001\_tso\_false}           & 519/3 &   n & \cp  & \bf 2.75  &   \textbf{ERR} &     & 252.8 & 86812  & 209  & 23.0 & 15453 	 &   1.9 \\\rowcolor{taskyblue!10}
\texttt{podwr000\_power\_opt\_false}  & 242/3 &   n & \cp  & 5.6  &   \textbf{ERR} &     & 0.8 & 4740  & 0.2  & \bf 0.7 & 2423 	 &   0.1 \\\rowcolor{white}
\texttt{thin001\_tso\_true.c}         &    194/3   &   y & \cp  & \bf 1.1  &   \textbf{ERR} &     & 8.5 &  69961 &  2.9  & 13.3 & 30267 & 1.6 \\\rowcolor{taskyblue!10} 
\hline
\end{tabular}
}
\vspace{0.25cm}
\caption{Overview of benchmarks.
The best time for each benchmark is in {\bf bold font}.
Results of \prog{Impara} were obtained with SVN version 866.\label{table:svcomp}
\textbf{WA} means that the tool produced a wrong alarm for a safe example.
\textbf{WP} means that the tool produced a wrong proof for an unsafe example.
}
\end{table*}

We ran our experiments on a 64-bit machine with a 3\,GHz Xeon processor. 
Table~\ref{table:svcomp} gives an overview of the results.
For each benchmark, we give the number of lines (LOC)
and the number of threads.
For \prog{CBMC} and \prog{Threader}, we give the running time
and a tick mark if the benchmark was solved successfully.
For \prog{AbPress} and \prog{FMCAD13}, we provide
the running time, the number of nodes $|V|$ in the ART, and
the time spent for solving SMT queries.
The \textbf{effectiveness of summarisation} is tested by switching
summarisation off, and, instead, enumerating the set of paths represented by
the summaries.  Our experiments confirm that summarisation dramatically
reduces the cost of dependency analysis.

\begin{wrapfigure}[16]{r}{0.5\linewidth}
\centering
\vspace{-.5cm}

\begin{tikzpicture}[scale=0.7] 
\begin{axis}[%
  xmin=.1,xmax=1500, ymin=.1, ymax=1500,
  xmode=log,%
  ymode=log,%
  xlabel={\prog{FMCAD'13}},ylabel={\prog{AbPress}},
  scatter/classes={%
		t={mark=square*,blue},%
		f={mark=triangle*,red},%
		e={mark=o,draw=black}}]] 

\addplot [domain=.1:1500] {x};
\addplot [red,sharp plot, domain=.1:1500] {900}
          node [below] at (axis cs:10,850) {timeout};
\addplot [red,sharp plot, domain=.1:1500] coordinates{(900,.1) (900,1500)}
          node [left,rotate=90] at (axis cs:700,150) {timeout};
	
  	  \addplot[scatter,only marks,%
				  	   scatter src=explicit symbolic]%
  	  				 table[meta=label] 
  	  				 {scatter1.txt};
	\end{axis}
\end{tikzpicture} 

\caption{AbPress vs.~FMCAD'13\label{scatter1}}
\end{wrapfigure}
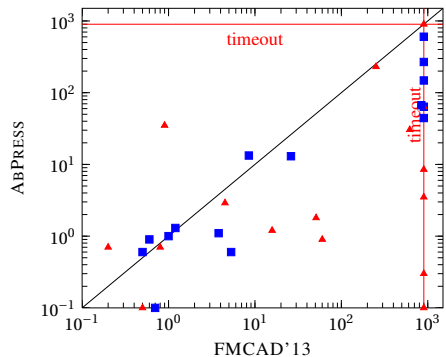
Without summarisation, we observe an order-of-magnitude increase of the
number of paths explored in dependency analysis compared to summarisation. 
For example, in \texttt{qrcu\_true} covering nodes have on average of around
14 postfixes.  This means on average 14 paths would have to be analysed
every time a cover between nodes is detected.  As a result, successful cover
checks become expensive.  However, the \prog{Impact} algorithm relies on
covers being both efficient to check and to undo.  In practice, this leads
to timeouts, primarily, in programs with loops.  For example, the analysis
of \texttt{qrcu\_true}, and \texttt{stack\_true} timed out after 900\,s. 
Overall, this naive algorithm is not competitive with \prog{FMCAD'13}.

We evaluate the \textbf{benefits of Source-sets versus peephole
partial-order reduction} by comparing against \prog{FMCAD'13}. 
Figure~\ref{scatter1} shows a scatterplot comparing the running times of
\prog{FMCAD'13} with \prog{AbPress}.  The latter is clearly superior,
resulting in both overall best running times and fewer timeouts.

As shown by Table~\ref{table:svcomp}, the number of ART nodes explored by
\prog{AbPress} is lower than for \prog{FMCAD'13}, except in unsafe
instances.  As peephole POR explores more interleavings, it may by chance
explore an interleaving with a bug earlier.

To evaluate the \textbf{competitiveness of AbPress}, as well as its
limitations, we have aimed to carry out a comprehensive evaluation, where we
deliberately retain examples that are not main strengths of \prog{AbPress},
e.g., where the number of threads is high or a very large number of thread
interleavings is required to expose bugs.

\prog{AbPress} \textbf{solves 10 SV-COMP benchmarks not solved by} \prog{Threader}.
Two of those \texttt{qrcu\_ok\_safe} and \texttt{qrcu\_false} are cases
where \prog{Threader} times out. The other cases are errors where our implementation
seems to be more robust in handling arrays and pointers\footnote{
\prog{AbPress} gives only one wrong result, as it currently does not take 
failure of memory allocation into account, which affects example \texttt{bigshot\_s\_false}.}.
Here the path-based nature of our algorithm can play out its strength
in determining aliasing information.
Furthermore, \prog{AbPress} is capable of dealing with
the weak-memory examples, where \prog{Threader} gives no results.

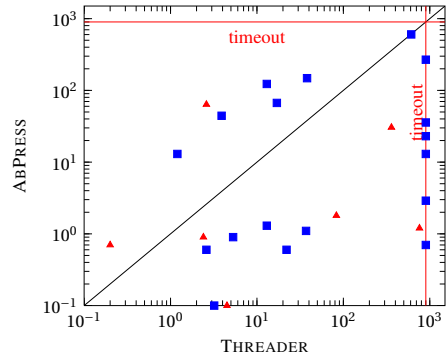
\begin{wrapfigure}[12]{r}{0.5\linewidth}
\centering
\vspace{-.5cm}
\begin{tikzpicture}[scale=0.7] 
\begin{axis}[%
  xmin=.1,xmax=1500, ymin=.1, ymax=1500,
  xmode=log,%
  ymode=log,%
  xlabel={\prog{Threader}},ylabel={\prog{AbPress}},
  scatter/classes={%
		t={mark=square*,blue},%
		f={mark=triangle*,red},%
		e={mark=o,draw=black}}]] 

\addplot [domain=.1:1500] {x};
\addplot [red,sharp plot, domain=.1:1500] {900}
          node [below] at (axis cs:10,850) {timeout};
\addplot [red,sharp plot, domain=.1:1500] coordinates{(900,.1) (900,1500)}
          node [left,rotate=90] at (axis cs:700,150) {timeout};


    \addplot[scatter,only marks,%
    scatter src=explicit symbolic]%
    table[meta=label] {scatter2.txt};
	\end{axis}
\end{tikzpicture} 
\caption{AbPress vs.~Threader\label{scatter2}}
\end{wrapfigure}
Disregarding the special \texttt{pthread-ext} category,
the only cases where \prog{AbPress} fails while \prog{Threader}
succeeds are the pathological Fibonacci examples and 
the indexer example, which features 14 threads.
Figure~\ref{scatter2} compares the running times of \prog{Threader}
with that of \prog{AbPress}.
The dots above the diagonal, where \prog{Threader} wins,
are mainly in the \texttt{pthread\_ext} category.

\section{Related Work}
Source-set based DPOR was recently presented in \cite{Abdulla:2014},
as part of state-less explicit state model-checker for Erlang programs.
While we borrow the notion of source-sets, our context is a
fundamentally different.
Hansen et al. consider a combination of partial-order reduction 
and zone abstractions for timed automata~\cite{DBLP:conf/cav/HansenLLN014}
where the dependence relation is computed from an abstract transformer.

Cimatti et al. \cite{DBLP:conf/tacas/CimattiNR11}
combine static POR with lazy abstraction to verify SystemC
programs. Our work differs from their work on multiple fronts:
SystemC has a significantly different concurrency
model than multi-threading, and we use an abstract dynamic POR, which is inherently more precise than static POR. 

We presented a combination of peephole partial-order with Impact
in \cite{wko2013}, however using peephole partial-order reduction
which is simpler to integrate than source sets
but leads to a greater number of interleavings, as demonstrated in our experiments.

\prog{Threader} is a software verifier
for multi-threaded programs~\cite{DBLP:conf/cav/GuptaPR11}
based on compositional reasoning and invariant inference by constraint solving.
In \cite{PRW2014}, Popeea et al present a combination
of abstraction for multi-threaded programs with Lipton's reduction.
Reduction is applied as a program transformation that inserts
atomic section based on a lockset analysis.
The authors then subsequently run \prog{Threader} on
the transformed program. Unfortunately, at its current stage,
their tool still requires
manual transformations, and therefore we did not test against this implementation.

\section{Conclusion}
We have presented a concurrent program 
model checking technique {\em AbPress} that incorporates an 
aggressive DPOR based on source-sets along  with 
Impact. Abstraction in the form of abstract summaries of shared accesses 
was utilized  
to amplify the effectiveness of 
DPOR with covers in the abstract reachability tree. We implemented
the {\em AbPress} algorithm in  Impara  and evaluated it against comparable
verifiers. Our initial results have been favorable. As a part of future work, we intend to 
use more aggressive property-guided abstractions to further reduce the interleaving space.

\bibliographystyle{splncs}
\bibliography{bib}

\end{document}